\documentclass[preprint]{imsart}
\usepackage[round]{natbib}

\usepackage{amsthm,amsmath,amssymb,graphicx, subcaption,natbib}
\usepackage[linesnumbered,ruled]{algorithm2e}
\RequirePackage[colorlinks,citecolor=blue,urlcolor=blue]{hyperref}
\usepackage{xr,paralist,booktabs} 
\usepackage{booktabs}
\startlocaldefs
\usepackage{paralist}

\newcommand{\tran}{\mathsf{T}}
\newcommand{\wh}{\widehat}

\newcommand{\diag}{\mathrm{diag}}

\renewcommand{\ge}{\geqslant}
\renewcommand{\le}{\leqslant}

\newcommand{\info}{\mathcal{I}}

\newcommand{\dnorm}{\mathcal{N}}
\newcommand{\sumdot}{\text{\tiny$\bullet$}}
\newcommand{\e}{\mathbb{E}}

\newcommand{\zij}{Z_{ij}}

\newcommand{\zipj}{Z_{i'j}}

\newcommand{\zzt}{(ZZ^\tran)}
\newcommand{\ztz}{(Z^\tran Z)}

\newcommand{\nid}{N_{i\sumdot}}
\newcommand{\ndj}{N_{\sumdot j}}

\newcommand{\nrd}{N_{r\sumdot}}
\newcommand{\nds}{N_{\sumdot s}}

\newcommand{\ai}{a_i}

\newcommand{\bj}{b_j}

\newcommand{\eij}{e_{ij}}

\newcommand{\ssa}{\sigma^2_A}
\newcommand{\ssb}{\sigma^2_B}
\newcommand{\sse}{\sigma^2_E}

\newcommand{\fpe}{\sigma^4_E}

\newcommand{\ssah}{\hat\sigma^2_A}
\newcommand{\ssbh}{\hat\sigma^2_B}
\newcommand{\sseh}{\hat\sigma^2_E}
\newcommand{\betah}{\hat\beta}

\newcommand{\ka}{\kappa_A}
\newcommand{\kb}{\kappa_B}
\newcommand{\ke}{\kappa_E}

\newcommand{\yij}{Y_{ij}}

\newcommand{\xij}{x_{ij}}

\newcommand{\xbid}{\bar x_{i\sumdot}}
\newcommand{\xbdj}{\bar x_{\sumdot j}}

\newcommand{\xbdd}{\bar x_{\sumdot \sumdot}}

\newcommand{\xtdj}{\tilde x_{\sumdot j}}

\newcommand{\var}{\mathrm{Var}}
\newcommand{\cov}{\mathrm{Cov}}
\newcommand{\eff}{\mathrm{eff}}

\newcommand{\simiid}{\stackrel{\mathrm{iid}}{\sim}}
\newcommand{\p}{\stackrel{p}{\to}}

\newcommand{\natu}{\mathbb{N}}
\newcommand{\real}{\mathbb{R}}

\newcommand{\vi}{v_{1,i}}
\newcommand{\vj}{v_{2,j}}

\newcommand{\phm}{\phantom{-}} 
\newcommand{\phe}{\phantom{=}\,\,} 

\newcommand{\ols}{\mathrm{OLS}}
\newcommand{\gls}{\mathrm{GLS}}
\newcommand{\rls}{\mathrm{RLS}}
\newcommand{\cls}{\mathrm{CLS}}

\newcommand{\bsu}{\boldsymbol{u}}
\newcommand{\bsx}{\boldsymbol{x}}

\newtheorem{thm}{Theorem}[section]

\theoremstyle{definition}

\newtheorem{model}{Model}  
\endlocaldefs

\begin{document}

\begin{frontmatter}

\title{Estimation and Inference for Very Large Linear Mixed Effects Models}
\runtitle{Very Large Linear Mixed Effects Models}


\begin{aug}
	\author{\fnms{Katelyn} \snm{Gao}\corref{}\ead[label=e1]{kxgao@stanford.edu}\thanksref{t1}}
	\and
	\author{\fnms{Art B.} \snm{Owen}\ead[label=e2]{owen@stanford.edu}\thanksref{t2}}
	
	\thankstext{t1}{Supported under US NSF Grant DGE-114747}
	\thankstext{t2}{Supported under US NSF Grant DMS-1407397}
	
	\runauthor{Gao and Owen}
	
	\affiliation{Stanford University}
	
	\address{Department of Statistics, Stanford University \\ 390 Serra Mall\\ Stanford, CA 94305\\ \printead{e1,e2}}

\end{aug}

\begin{abstract}
Linear mixed models with large imbalanced crossed random effects structures
pose severe computational problems for maximum likelihood estimation
and for Bayesian analysis. The costs can grow as fast as $N^{3/2}$ when there are
$N$ observations.  Such problems arise in any setting where the underlying
factors satisfy a many to many relationship (instead of a nested one)
and in electronic commerce applications, the $N$ can be quite large.
Methods that do not account for the correlation structure can greatly underestimate uncertainty.
We propose a method of moments approach that takes account of the correlation
structure and that can be computed at $O(N)$ cost.
The method of moments is very amenable to parallel computation
and it does not require parametric distributional assumptions, tuning parameters
or convergence diagnostics.
For the regression coefficients, we give conditions for consistency and
asymptotic normality as well as a consistent variance estimate.
For the variance components, we give conditions for consistency and
we use consistent estimates of a mildly conservative variance estimate.
All of these computations can be done in $O(N)$ work.
We illustrate the algorithm with some data from Stitch Fix where the crossed random effects
correspond to clients and items. 
\end{abstract}

\begin{keyword}[class=MSC]
\kwd{62F10}
\kwd{62K05}
\kwd{62K10}
\end{keyword}

\begin{keyword}
\kwd{linear mixed models}
\kwd{crossed}
\kwd{scalable}
\end{keyword}

\end{frontmatter}

\section{Introduction}\label{sec:intro}

The field of statistics is confronting two important challenges at present.
The first is the arrival of ever larger data sets, sometimes described as `big data'.
See, for instance,  \cite{prov:fawc:2013} and~\cite{vari:2014}.
The second is the reproducibility crisis, in which published findings cannot be replicated.
This problem was clearly presented by~\cite{ioan:2005}
among others, and it has lead to the American Statistical Association
releasing a statement on $p$-values \citep{wass:laza:2016}.

We might naively hope that the first problem would remove the second one.
While larger data sets can greatly reduce uncertainty, difficulties remain.
The one that we consider here is a crossed random effects structure in the data.
That structure introduces a dense tangle of correlations that can sharply
reduce the effective sample size of the data at hand.  If, as we suspect, most
data scientists treat these large data sets as IID samples, then they will
greatly underestimate the uncertainty in their fitted models.
The usual methods for this problem, whether based on maximum likelihood or on Bayes,
scale badly to large data sets, with a cost that grows superlinearly
in the sample size. We present and study a method of moments approach
with cost that scales linearly in the problem size, among other advantages.

The sort of data that motivate us arise in e-commerce applications.
The factors are variables like cookies, customer IDs, query strings,
IP addresses, product IDs (e.g., SKUs), URLs and so on.
The most direct way to handle such variables is to treat them as
categorical variables that simply happen to have a large number of
levels, including many that have not yet appeared in the data.
We think that a random effects model is more appropriate.
For instance, internet cookies are cleared regularly and hence any specific
cookie is likely to disappear shortly.   It is therefore appropriate
to consider the specific cookies in a data set as a sample from some
distribution, that is, as a random effect. Similarly there is turnover in
popular products and queries, that motivate treating them as random effects
too.

While the largest crossed random effect data sets we know of are in e-commerce, we expect
the problem to arise in other settings where data set sizes are growing. The crossed random
effects structure is very fundamental. 
Any setting with a many to many mapping of factor levels
involves crossed effects that one might want to model as random.
In agriculture there are gene by environment crosses. In education, neither schools
nor neighborhoods are perfectly nested within the other \citep{R93} and in multiyear
data sets there is a many to many relationship between teachers and students.

When our chosen model involves only one of these random effect entities 
then a hierarchical model, based on Bayes or empirical
Bayes, can be quite effective.
Things change considerably when we want to use two or more
crossed random effects.  
In this paper, we consider the following model,
\begin{model}\label{mod:twofactor}
Two-factor linear mixed effects:
\begin{equation}\label{eq:refmodel}
\begin{split}
&\yij=\xij^\tran\beta+a_i+b_j+e_{ij},\quad \xij\in\real^p,\quad i,j\in\natu\quad\text{where},\\
&a_i\simiid (0,\ssa),\quad b_j\simiid(0,\ssb),\quad e_{ij}\simiid(0,\sse)\quad\text{(independently) and}, \\
&\e(a_i^4) < \infty,\quad \e(b_j^4) < \infty,\quad \e(e_{ij}^4) < \infty. 
\end{split}
\end{equation}
\end{model}
For instance, customer $i$ might assign a score $Y_{ij}$ to product $j$.
Then $\xij$ contains features about the customer or product or
some joint properties of them, $\beta$ is of interest for the company 
choosing what product to recommend, $b_j$ measures some general
appeal of the product not captured by the features in $\xij$,
while $a_i$ captures variation in which customers are harder or easier
to please and $e_{ij}$ is an error term.
This is a mixed effects model because it contains both random effects
$\ai$, $\bj$ and fixed effects $\xij$.

Model~\ref{eq:refmodel} describes any $ij$ pair,  but the given data
set will only contain some finite number $N$ of them.
If the available data are laid out as rows $i$ and columns $j$
with $R$ distinct rows and $C$ distinct columns, then  the cost
of fitting a generalized least squares regression model
for $\beta$ scales as $O((R+C)^3)$ because it solves a
$p\times p$ system of equations with $p\ge R+C$.
See \cite{SCM92}, \cite{R93} and \cite{B14}.
Now because $RC\ge N$
we have $\max(R,C)\ge\sqrt{N}$ and $(R+C)^3>N^{3/2}$.

\cite{GO17} consider an intercept-only version of Model~\ref{eq:refmodel}
where $\xij^\tran\beta$ is simply a constant $\mu\in\real$ for all $i$ and $j$.
They find that Markov chain Monte Carlo does not solve the inference problem.
All of the MCMC methods considered either failed to mix, or converged to the wrong
answer, and this took place already at modest sample sizes. For the specific case of
a Gibbs sampler and Gaussian $\ai$, $\bj$ and $\eij$, 
using methods from \cite{RS97} they show it will take $O(N^{1/2})$ iterations
costing $O(N)$ each to converge, for a total cost of $O(N^{3/2})$.
\cite{F13} presents a very general equivalence between the convergence 
rate of an iterative equation solver and the convergence rate of an 
associated MCMC scheme, so these identical rates may be a sign of
a deeper equivalence.
Consensus Bayes \citep{scot:etal:2016} splits the data into shards, one
per processor.  However the data given to each shard has to be independent
and here data sets corresponding to a subset of rows will have correlations
due to commonly sampled columns (and vice versa).

We find that existing Bayes and likelihood methods are not effective
for this problem. Here we present an approach based on the method
of moments.  We seek estimates  $\hat\beta$, $\ssah$, $\ssbh$ and
$\sseh$ along with variance estimates for these quantities.
We have three criteria:
\begin{compactenum}[\quad1)]
\item the total computational cost {\sl must} be $O(N)$ time and $O(R+C)$ space,
\item the variance estimates {\sl should} be reliable or conservative, and
\item we {\sl prefer} $\hat\beta$ to be statistically efficient.
\end{compactenum}
We regard the first criterion as a constraint that must be met.
For the second criterion, a mild over-estimate of $\var(\hat\beta)$ is
acceptable in order to keep the costs in $O(N)$.  
The third criterion is to be met as well as we can, subject
to constraints given by the first two. Computational efficiency
is more important than statistical efficiency in this context. 
For very large $N$,
requiring $O(N^{3/2})$ computation is like asking for an oracle.

The method of moments meets our $O(N)$ time and $O(R+C)$ space
criteria, and here we show that it can also yield reliable variance
estimates.  Further advantages of the method of moments are
that it does not require parametric distributional assumptions, there are
no tuning parameters to choose, and most importantly for large $N$, it
is very well suited to parallel computation. The method of moments is not
without drawbacks.  Sometimes it yields parameter estimates that are out of
bounds, such as negative variance estimates.  

An outline of this paper is as follows. 
Section~\ref{sec:background} introduces most of the notation for Model~\ref{mod:twofactor},
especially the pattern of missingness in the data, and gives some of the asymptotic assumptions.
Section~\ref{sec:algorithm} presents our algorithm 
and shows that it takes $O(N)$ time and $O(R+C)$ space.
We  compute a generalized least squares (GLS) estimate for a model
with either row or column variance components, but not both. 
We choose based on an efficiency criterion.
Then we estimate $\var(\betah)$ accounting for all three error terms including the one
left out of the GLS estimate.
Section~\ref{sec:stitchfix} illustrates our algorithm with some ratings data from Stitch Fix.  
There $\yij$ is a rating, from a $10$ point scale, by customer $i$ on item $j$, with features $\xij$.
Compared to OLS estimates, the random effects model leads to standard errors on coefficients
 $\beta_j$ that can be more than ten times higher.   That may be interpreted as an effective sample size
which is less than $1$\% of the nominal sample size.
Section~\ref{sec:theory} gives conditions under which $\betah\p\beta$,
and further conditions for  the variance components to be consistent.
There is also a central limit theorem for $\betah$.
Section~\ref{sec:experiments} compares our method of moments estimator
to a state of the art GLMM code \citep{lme4julia} written in Julia
\citep{beza:edel:karp:shah:2017}.
We find that algorithm takes $O(N^{3/2})$ cost per iteration, with
a number of iterations that, in our simulations, depends on $N$.
On problems where the GLMM code gives an answer we find it more
statistically efficient for $\beta$ and $\sse$ but not for $\ssa$ or $\ssb$.
Section~\ref{sec:conclusion} discusses some future work. 

Our method of moments approach is similar to methods \cite{H53} develops for Gaussian data. 
\cite{GO17} use $U$-statistics to find a counterpart to the Henderson I estimator that can be computed in $O(N)$ time and $O(R+C)$ space.  
They also get a variance estimator for their variance components, without
assuming a Gaussian distribution.
The variance estimator can be computed in $O(N)$ time.
It targets a mildly conservative upper bound on the variance
as the variance itself, like the one for Henderson's estimates, takes more than $O(N)$ computation.
In this paper we incorporate fixed effects along with the random effects, just as Henderson II 
does in generalizing Henderson I.  Henderson III allows for interactions between fixed and random effects.
We believe such interactions are very reasonable 
in our motivating applications, but incorporating them is outside the scope of this article.

Our analysis is for a fixed dimension $p$. This is reasonable for our
motivating data from Stitch Fix, where $p\ll N$.
It remains to develop methods for cases where $p\to\infty$ with $N$.

Another issue that we do not address in this article is selection  bias
in the available observations. 
Sometimes ratings are biased towards the high end because
customers seek products that they expect to like and companies endeavor to 
recommend such
products. In other data sets, such as restaurant reviews, customers
may be more likely to make a rating when they are  either very unhappy 
or very happy.  For such data, the
ratings will be biased towards both extremes and away from the middle.
Accounting for selection bias requires assumptions or information from outside 
the given data.
Propensity weighting \citep[Chapter 13]{imbe:rubi:2015} could well
fit into our framework, but we leave that out of this paper, as the basic
problem without selection bias is already a challenge.

\section{Notation and asymptotic conditions}\label{sec:background}

Here we give a fuller presentation of our notation.
Equation \eqref{eq:refmodel} describes the distribution
of seen and future data.
We call the first index of $\yij$ the `row' and the second the `column'. We use integers $i,i',r,r'$ to index rows and $j,j',s,s'$ for columns, but the actual indices may be URLs, customer IDs, or query strings.
The index sets are countably infinite to always 
leave room for unseen levels in the future.

The variable $\zij$ takes the value $1$ if $(\xij,\yij)$ is observed and $0$ otherwise. 
We assume that there is at most one observation in position $(i,j)$. 
For customer rating data, we suppose that if $i$ has rated $j$ multiple times, then
only the most recent rating is retained.  In many other settings, only a negligible
fraction of $ij$ pairs will have been duplicated.

The sample size is $N=\sum_{ij}\zij<\infty$. The number of observations in row $i$ is $\nid=\sum_j\zij$ and the number in column $j$ is $\ndj=\sum_i\zij$. The number of distinct rows is $R=\sum_i1_{\nid>0}$ and there are $C=\sum_j1_{\ndj>0}$ distinct columns. In the following, summing over rows $i$ means
summing over just the $R$ rows $i$ with $\nid>0$, and similarly for sums over columns. 
This convention corresponds to what happens when one makes a pass through the whole data set. 

Let $Z$ be the matrix containing $\zij$. 
Then $\zzt_{ii'} = \sum_j\zij\zipj$ is the number of columns for which we have data in both rows 
$i$ and $i'$. Similarly, $\ztz_{jj'}$ is the number of rows in which both columns $j$ and
 $j'$ are observed. Note that $\zzt_{ii'}\le\nid$ and $ \ztz_{jj'}\le\ndj$.
We will use the following identities:
$$\sum_{ir}\zzt_{ir} =\sum_j\ndj^2,\quad\text{and}\quad 
\sum_{js}\ztz_{js}=\sum_i\nid^2.$$

This notation allows for an arbitrary pattern of observations. We mention three special cases. A 
balanced crossed design has $\zij = 1_{i\le R}1_{j\le C}$. If $\max_i\nid=1$ but 
$\max_j\ndj>1$ then the data have a hierarchical  structure with rows nested in columns. 
If $\max_i\nid=\max_j\ndj=1$, then the observed $\yij$ have IID errors. Some of these patterns 
cause problems for parameter estimation. For example, if the errors are IID, then the variance 
components are not identifiable. Our assumptions rule these out to focus on large genuinely crossed data sets.

The following vectors are useful for subsequent analyses. Let $\vi$ be the length-$N$ vector 
with ones in entries $\sum_{r=1}^{i-1} \nrd+1$ to $\sum_{r=1}^i \nrd$ and zeros elsewhere. 
Similarly, let $\vj$ be the length-$N$ vector with ones in entries $\sum_{s=1}^{j-1}\nds+1$ 
to $\sum_{s=1}^j \nds$ and zeros elsewhere.


Next, we describe our asymptotic assumptions. 
 First
\begin{align}\label{eq:assum:nohogs}
\epsilon_R = \max_i\nid/N \to 0,\quad\text{and}\quad\epsilon_C=\max_j\ndj/N \to 0,
\end{align}
so no single row or column dominates.
The average row size can be measured by $N/R$ or by 
$\sum_i\nid^2/N$; the latter is $\e(\nid)$ when choosing one of the $N$ data points 
$(i,j,\xij,\yij)$ at random (uniformly).
Similar formulas hold for the average column size. 
These average row and column sizes are $o(N)$, because
$$\frac1{N^2}\sum_i\nid^2\le\epsilon_R\to0,\quad
\text{and}\quad \frac1{N^2}\sum_j\ndj^2\le\epsilon_C\to0.$$
We often expect the average row and column sizes, while growing slower than $N$,
should diverge:
\begin{align*}
&\min(N/R,N/C) \to \infty,\quad\text{and}\\
\begin{split}
&\min\Bigl(\frac1N\sum_i\nid^2,\frac1N\sum_j\ndj^2\Bigr) \to \infty.
\end{split}
\end{align*} 
We do not however impose these conditions.


Even for large average row and columns sizes, there can still be numerous new
or rare entities with $\nid=1$ or $\ndj=1$.
Our analysis can include such small rows and columns without requiring
that they be deleted.
When there are covariates $\xij$ we need to rule out
degenerate settings where the sample variance of $\xij$ does not grow with $N$
or where it is dominated by a handful of observations.  We add some such
conditions when we discuss central limit theorems in Section~\ref{sec:clt}.

The finite fourth moments $\e(a_i^4)$, $\e(b_j^4)$ and $\e(e_{ij}^4)$
are conveniently described through finite kurtoses $\ka$, $\kb$ and $\ke$,
respectively.  Some of the variance expressions
in \cite{GO17} are dominated by terms proportional to $\kappa+2$
for one of these kurtoses.
Following \cite{GO17} we assume that $\min(\ka,\kb,\ke)>-2$.
This lower bound rules out some symmetric binary distributions
for $\ai$, $\bj$ and $\eij$. Such cases seem unrealistic for our
motivating applications.

The randomness in $\yij$ comes from $\ai$, $\bj$ and $\eij$.
In some places we combine them into $\eta_{ij} \equiv \ai+\bj+\eij$.

\cite{GO17} found exact finite sample formulas 
for the variance of their method
of moments estimators $\ssah$, $\ssbh$ and $\sseh$. 
They then derived asymptotic expressions 
letting $\epsilon_R$, $\epsilon_C$, $R/N$ and $C/N$
approach zero.  The Stitch Fix data 
that we consider in Section~\ref{sec:stitchfix}
does not have a very small value for $R/N$.  Here we 
develop non-asymptotic magnitude bounds for 
bias and variance that do not 
require $R/N$ and $C/N$ to be close to zero. 
They need only be bounded away from one. 

\begin{thm}\label{thm:varcoefbiasandvar}
Suppose that $\max(R/N,C/N)\le\theta$
for some $\theta<1$ and let $\epsilon=\max(\epsilon_R,\epsilon_C)$. 
Then the moment based estimators from~\cite{GO17} satisfy
\begin{align*}
\e(\ssah)& =
(\ssa +\Upsilon)  (1+O(\epsilon)),\\
\e(\ssbh) & =(\ssb+\Upsilon) (1+O(\epsilon)),\quad\text{and}\\
\e(\sse) & =(\sse+\Upsilon)(1+O(\epsilon)),
\end{align*}
where 
$$\Upsilon \equiv 
\ssa \frac{\sum_i\nid^2}{N^2}
+\ssb\frac{\sum_j\ndj^2}{N^2}
+\frac{\sse }N = O(\epsilon).
$$
Furthermore
\begin{align*}
\max\Bigl(
\var(\ssah),\var(\ssbh),\var(\sseh)\Bigr) 
 & = 
O\Bigl( \frac{\sum_i\nid^2}{N^2}+\frac{\sum_j\ndj^2}{N^2}\Bigr) 
=O(\epsilon). 
\end{align*}
\end{thm}
\begin{proof}
See Section~\ref{sec:backgroundproof} in the supplement.
\end{proof}

Theorem~\ref{thm:varcoefbiasandvar} has the same variance
rate for all variance components. In our computed examples $\var(\sse)\ll\min(\var(\ssa),\var(\ssb))$
because $N\gg\max(R,C)$, a condition not imposed in Theorem~\ref{thm:varcoefbiasandvar}.
Both bias and variance are  $O(\epsilon)$ and so
a (conservative) effective sample size is then $O(1/\epsilon)$.
The quantity $\Upsilon$ appearing in Theorem~\ref{thm:varcoefbiasandvar}
is $\var(\bar Y_{\sumdot\sumdot})$
where $\bar Y_{\sumdot\sumdot} = (1/N)\sum_{ij}\zij\yij$.
The variances of the variance components contain similar quantities to $\Upsilon$
although kurtoses and other quantities appear in their implied constants.

\section{An Alternating Algorithm}\label{sec:algorithm}

Our estimation procedure for Model~\ref{eq:refmodel} is given in
Algorithm~\ref{alg:alternating}. We alternate twice between finding $\betah$ and the variance component estimates $\ssah$, $\ssbh$, and $\sseh$. Further details of these steps, including the way we choose generalized least squares (GLS) estimator to use in step $3$, are given in the next two subsections.

The data are a collection of  $(i,j,\xij,\yij)$ tuples. A pass over the data proceeds via iteration over all tuples in the data set. Such a pass may generate $O(R+C)$ intermediate values to be retained for future computations. 

\begin{algorithm} [ht] 
	Estimate $\beta$ via ordinary least squares (OLS): $\betah=\betah_{\ols}$. \\
	Let $\ssah$, $\ssbh$, and $\sseh$ be the method of moments estimates from \citep{GO17} defined on the data $
	(i,j,\hat\eta_{ij})$, where $\hat\eta_{ij} = \yij-\xij^\tran\betah_{\ols}$.\\
	Compute a more efficient $\betah$ using $\ssah$, $\ssbh$, and $\sseh$. If $\ssah\max_i\nid\ge\ssbh\max_j\ndj$, 
	estimate $\beta$ via GLS accounting for row correlations: $\betah=\betah_{\rls}$. Otherwise, estimate it via GLS 
	accounting for column correlations: $\betah=\betah_{\cls}$. \\
	Repeat step 2 using $\hat\eta_{ij} = \yij-\xij^\tran\betah$ with $\betah$ from step 3.\\
Compute an estimate $\wh\var(\hat\beta)$ for $\betah$ from step 3 using $\ssah$, $\ssbh$ and $\sseh$ from step 4.
	\caption{Alternating Algorithm}
	\label{alg:alternating}
\end{algorithm}

\subsection{Step by step details for Algorithm~\ref{alg:alternating}}\label{sec:rlscls}

\subsubsection{Step 1}\label{sec:step1}

The first step of Algorithm~\ref{alg:alternating} is to compute the OLS estimate of $\beta$. 
Let $X\in\real^{N\times p}$ have rows $\xij$ in some order and let
$Y\in\real^{N}$ be elements $\yij$ in the same order. Then,
\begin{align}\label{eq:ols}
\betah_{\ols} &= (X^\tran X)^{-1}X^\tran Y = \biggl(\sum_{ij}\zij\xij\xij^\tran\biggr)^{-1}\sum_{ij}\zij\xij\yij.
\end{align}

In one pass over the data, we can compute $X^\tran X$ and $X^\tran Y$ and solve for $\betah$. 
Solving the normal equations this way is easy to parallelize but more prone to roundoff error than 
the usual alternative based on computing the SVD of $X$.  
The numerical conditioning of the SVD computation 
is like doubling the number of floating point bits available compared to 
solving normal equations.
One can compensate by solving normal equations in extended precision.
It costs $O(p^3)$ to compute $\betah_\ols$ and so the cost of step $1$ is 
$O(Np^2+p^3)$. The space cost is $O(p^2)$.

\subsubsection{Step 2}\label{sec:step2}

Step 2 uses the algorithm from \cite{GO17} to compute variance component estimates
$\ssah$, $\ssbh$ and $\sseh$ in $O(N)$ time and $O(R+C)$ space.
A more detailed account is in Section~\ref{sec:mom}.
This takes $O(Np)$ time to recompute~$\hat\eta_{ij}$.

\subsubsection{Step 3}\label{sec:step3}

\paragraph{GLS estimators:} 
First we define and compare  GLS estimators of $\beta$ accounting for row correlations, or column 
correlations, or both.  These estimators are most easily presented through a reordering of the data. 
Our algorithm does not have to actually sort the data which would be
a major inconvenience in our motivating applications.
We work with one row ordering of the data, in which $ij$ precedes $i'j'$ whenever $i<i'$
and with one column ordering of the data.
Let $P$ be the $N\times N$ permutation matrix corresponding to the transformation of the column ordering to 
the row ordering. Let $A_R\in\natu^{N\times N}$ be the block diagonal matrix with $i$'th block $1_{\nid}1_{\nid}^\tran$ and $B_C\in\natu^{N\times N}$  the
 block diagonal matrix with $j$'th block $1_{\ndj}1_{\ndj}^\tran$. 

If $Y$ is given in the row ordering, then
\begin{align}\label{eq:rowcov}
\cov(Y) & = V_R \equiv \sse I_N+\ssa A_R+\ssb B_R, \quad\text{for}\quad B_R=P B_C P^\tran.
\end{align}
For $Y$ in the column ordering,
\begin{align}\label{eq:colcov}
\cov(Y) & = V_C \equiv \sse I_N+\ssa A_C+\ssb B_C, \quad\text{for}\quad A_C=P^\tran A_R P.
\end{align}
GLS algorithms
based on~\eqref{eq:rowcov} or~\eqref{eq:colcov} have computational complexity $O(N^{3/2})$.
This is better than the $O(N^3)$ that we might face had $V_R$ or $V_C$ been arbitrary
dense matrices, instead of being comprised of the identity and some low rank block
diagonal matrices, but it is still too slow for large scale applications.

In a hierarchical model where only row correlations were present we could
take $\ssb=0$ and define 
\begin{align}\label{eq:rls}
\betah_{\rls} &= (X^\tran \hat V_A^{-1} X)^{-1}X^\tran\hat V_A^{-1}Y, \quad\text{for}\quad
 \hat V_A=\sseh I_N+\ssah A_R,
\end{align}
using sample estimates $\ssah$ and $\sseh$ of $\ssa$ and $\sse$.
This GLS estimator of $\beta$ accounts for the intra-row correlations in the data.
Similarly, the GLS estimator of $\beta$ accounting for the intra-column correlations is 
\begin{align}\label{eq:cls}
\betah_{\cls} &= (X^\tran \hat V_B^{-1} X)^{-1}X^\tran \hat V_B^{-1}Y, \quad \text{for}\quad
\hat V_B = \sseh I_N+\ssbh B_C.
\end{align}
We show next that  $\betah_{\rls}$ and $\betah_{\cls}$ can be computed in $O(N)$ time.


\paragraph{GLS Computations in $O(N)$ cost:} 
From the Woodbury formula \citep{H89} and defining $Z_a\in\{0,1\}^{N\times R}$ 
as the matrix with $i$th column $\vi$
(from Section~\ref{sec:background}), we have
\begin{align*}
X^\tran \hat V_A^{-1}X &= X^\tran (\sseh I_N+\ssah Z_a{Z_a}^\tran)^{-1} X \notag \\
&= \dfrac{X^\tran X}{\sseh}-\dfrac{\ssah}{\sseh}X^\tran Z_a\, \diag\Bigl(\dfrac{1}{\sseh+\ssah\nid}\Bigr)Z_a^\tran X \notag \\
&= \dfrac{1}{\sseh}\sum_{ij}\zij\xij\xij^\tran-\dfrac{\ssah}{\sseh}\sum_i\dfrac{1}{\sseh+\ssah\nid}
\biggl(\sum_j\zij\xij\biggr)\biggl(\sum_j\zij\xij\biggr)^\tran.
\end{align*}
Likewise, $X^\tran\hat V_A^{-1}Y$
equals 
\[\dfrac{1}{\sseh}\sum_{ij}\zij\xij\yij-\dfrac{\ssah}{\sseh}\sum_i\dfrac{1}{\sseh+\ssah\nid}
\biggl(\sum_j\zij\xij\biggr)\biggl(\sum_j\zij\yij\biggr).\]

One pass over the data allows us to compute $\sum_{ij}\zij\xij\xij^\tran$ and $\sum_{ij}\zij\xij\yij$, as well as $\nid$, and the row sums $\sum_j\zij\xij$ and $\sum_j\zij\yij$ for $i=1,\dots,R$. The cost is $O(Np^2)$ time and $O(Rp)$ space. None of these quantities require us to sort the data. We then compute $X^\tran \hat V_A^{-1}X$ and $X^\tran \hat V_A^{-1}Y$ in time $O(Rp^2)$. Then, $\betah_{\rls}$ is computed in $O(p^3)$. Hence, $\betah_{\rls}$ can be found within $O(Rp)$ space and  $O(Np^2+p^3)=O(Np^2)$ time.
Clearly $\betah_{\cls}$ costs $O(Cp)$ space and $O(Np^2)$ time.


\paragraph{Efficiencies:} 

We can compute either  $\betah_{\rls}$ or $\betah_{\cls}$ in our computational budget.  
We will choose RLS if the variance component associated with rows is dominant
and CLS otherwise.  The choice could be made dependent on $X$ but in many
applications one considers numerous different $X$ matrices and we prefer
to have a single choice for all regressions.  Accordingly, we find a lower bound on the efficiency
of RLS when $X$ is a single nonzero vector $\bsx\in\real^{N\times 1}$.  We choose RLS if that lower bound is higher
than the corresponding bound for CLS, in this $p=1$ setting.

The full GLS estimator is
$\betah_{\gls} = (X^\tran V_R^{-1}X)^{-1}X^\tran V_R^{-1}Y$
 when the data are ordered by rows and $(X^\tran V_C^{-1}X)^{-1}X^\tran V_C^{-1}Y$ when the data are ordered by columns. 
For data ordered by rows, the efficiency of $\betah_{\rls}$ is
\begin{align}\label{eq:effrls}
\eff_{\rls} &= \dfrac{\var(\betah_{\gls})}{\var(\betah_{\rls})} = \dfrac{(\bsx^\tran V_A^{-1}\bsx)^2}{(\bsx^\tran V_A^{-1}V_RV_A^{-1}\bsx)(\bsx^\tran V_R^{-1}\bsx)}.
\end{align}
For data ordered by columns, the corresponding efficiency of $\betah_{\cls}$ is
\begin{align}\label{eq:effcls}
\eff_{\cls} &= \dfrac{\var(\betah_{\gls})}{\var(\betah_{\cls})} = \dfrac{(\bsx^\tran V_B^{-1}\bsx)^2}{(\bsx^\tran V_B^{-1}V_CV_B^{-1}\bsx)(\bsx^\tran V_C^{-1}\bsx)}.
\end{align}
The next two theorems establish lower bounds on these efficiencies.

\begin{thm}\label{thm:kantorovich}
Let $A$ be a positive definite Hermitian matrix and $\bsu$ be a unit vector. 
If the eigenvalues of $A$ are bounded below by $m>0$ and above by $M<\infty$, then
\[(\bsu^\tran A\bsu)(\bsu^\tran A^{-1}\bsu)\le\dfrac{(m+M)^2}{4mM}.\]
Equality may hold, for example when $\bsu^\tran A\bsu=(M+m)/2$ and the only roots of $A$ are $m$ and $M$.
\end{thm}
\begin{proof}
This is Kantorovich's inequality \citep{MO90}.
\end{proof}

By two applications of Theorem~\ref{thm:kantorovich} on \eqref{eq:effrls} and \eqref{eq:effcls}
we prove:
\begin{thm}\label{thm:worstcaseeff}
For $p=1$ and $\sse>0$, let $\eff_{\rls}$ and $\eff_{\cls}$ be defined as in \eqref{eq:effrls} and \eqref{eq:effcls}. Then
\begin{align*}
\eff_{\rls} &\ge \dfrac{4\sse(\sse+\ssb\max_j\ndj)}{(2\sse+\ssb\max_j\ndj)^2} \quad\text{and}\\
\eff_{\cls} &\ge \dfrac{4\sse(\sse+\ssa\max_i\nid)}{(2\sse+\ssa\max_i\nid)^2}.
\end{align*}
Both inequalties are tight.
\end{thm}
\begin{proof}
See Section~\ref{sec:proofworsteff} in the supplement.
\end{proof}

After some algebra, we see that the worst case efficiency of $\betah_{\rls}$ is higher than that of $\betah_{\cls}$ when $\ssa \max_i \nid > \ssb \max_j \ndj$. We set $\betah$ to be $\betah_{\rls}$ when  $\ssah\max_i\nid \ge \ssbh\max_j\ndj$, and $\betah_{\cls}$ otherwise. 

Optimizing a lower bound does not necessarily optimize the quantity
of interest, and so we expect that our choice here is not the
only reasonable one.
The efficiency of $\betah_{\rls}$ depends only on the ratio $\ssah/\sseh$
in use. We investigated GLS estimators of $\beta$ based on 
$\hat V_A=\ssah A_R +(\sseh+\lambda\ssbh)I_N$ for $\lambda$
chosen by the Kantorovich inequality.
It did not appear to bring an improved accuracy over our default choice
in some simulations.
In practice, one can also compute both $\betah_{\rls}$ and
$\betah_{\cls}$ and compare $\wh\var(\betah_{\rls})$
and  $\wh\var(\betah_{\cls})$.

\subsubsection{Steps 4 and 5}\label{sec:step4}
Step 4 is just like step 2 and it costs $O(Np)$ time.
Step 5 is described in Section~\ref{sec:varbetahrls} where
we derive $\var(\betah_\rls)$ and $\var(\betah_\cls)$.

\subsection{Method of Moments (Steps 2 and 4)}\label{sec:mom}

In this subsection, we discuss steps $2$ and $4$ of Algorithm~\ref{alg:alternating} in more detail. The errors $\yij-\xij^\tran\beta$ follow a two-factor crossed random effects model \citep{GO17}. If $\betah$ is a good estimate of $\beta$, then the residuals $\hat\eta_{ij}=\yij-\xij^\tran\betah$ approximately follow a two-factor crossed random effects model with $\mu=0$ and variance components $\ssa$, $\ssb$, and $\sse$. 

We estimate $\ssa$, $\ssb$, and $\sse$, with the algorithm from \cite{GO17} with data $(i,j,\hat\eta_{ij})$. That algorithm gives unbiased estimates of the variance components in a two-factor crossed random effects model. 

The algorithm of \cite{GO17} applies the method of moments to three statistics; 
a weighted sum of within-row sample variances, a weighted sum of within-column 
sample variances, and a multiple of the full sample variance.
For Algorithm~\ref{alg:alternating}, these are:
\begin{align}\label{eq:Ustats}
\begin{split} 
U_a(\betah) & = \sum_i S_{i\sumdot},\quad S_{i\sumdot} = \sum_j\zij(\hat\eta_{ij}-
\overline{\hat\eta}_{i\sumdot})^2\\
U_b(\betah) &= \sum_j S_{\sumdot j},\quad S_{\sumdot j} = \sum_i\zij(\hat\eta_{ij}-
\overline{\hat\eta}_{\sumdot j})^2,\quad\text{and}\\
U_e(\betah) &= \sum_{ij}\zij(\hat\eta_{ij}-\overline{\hat\eta}_{\sumdot\sumdot})^2,
\end{split}
\end{align}  
where subscripts replaced by $\sumdot$ are averaged over.
The variance component estimates are obtained by solving the system
\begin{align}\label{eq:mom} 
M\begin{pmatrix}\ssah \\ \ssbh \\ \sseh \end{pmatrix} &= \begin{pmatrix}U_a(\betah) \\ U_b(\betah) \\ U_e(\betah) \end{pmatrix},\  
M=\begin{pmatrix}0 & N-R & N-R \\ N-C & 0 & N-C \\ N^2-\sum_i \nid^2 & N^2-\sum_j \ndj^2 & N^2-N\end{pmatrix}.  
\end{align} 
The matrix $M$ is nonsingular under very weak conditions.
It suffices to have $R\ge2$, $C\ge2$, $\epsilon_R\le1/2$ and $\epsilon_C\le1/2$
\citep[Section 4.1]{GO17}.

\cite{GO17} compute the $U$-statistics in one pass over the data taking $O(N)$ time and $O(R+C)$ space. Solving \eqref{eq:mom} takes constant time. Thus, steps $2$ and $4$ each have computational complexity $O(N)$ and space complexity $O(R+C)$.

\section{Stitch Fix rating data}\label{sec:stitchfix}


Stitch Fix sells clothing, mostly women's clothing.  They mail their clients a sample of clothing items.
A client keeps and purchases some items and returns the others.   It is important to predict which items a client will like. In the context of our model, client $i$ might get item $j$ and then rate that item with a score $Y_{ij}$. 

Stitch Fix has provided us some of their client ratings data. This data is fully anonymized and void of any personally identifying information. The data  provided by Stitch Fix is a sample of their data, and consequently does not reflect their actual numbers of clients, items or their ratios, for example. Nonetheless this is an interesting data set with which to illustrate a linear mixed effects model. 

We received data on clients' ratings of items they received, as well as the following information about the clients and items. For client $i$ and item $j$, the response is a composite rating $Y_{ij}$ on a scale from $1$ to $10$.
There was a categorical variable giving the item's material. 
We also received a binary variable indicating whether the item style is considered to be `edgy', and another one on whether the client likes edgy styles. Similarly, there was another pair of binary variables indicating whether items were labeled `boho' (Bohemian) and whether the client likes boho items. 
 Finally, there was a match score. That is an estimate of the probability that the client keeps the item, predicted before it is actually sent. The match score is a prediction from a baseline model and is not representative of all algorithms used at Stitch Fix. 

The observation pattern in the data is as follows.
We received $N=5{,}000{,}000$ ratings on $C=6{,}318$ items by $R=762{,}752$ clients. Thus $C/N\doteq 0.00126$ and $R/N\doteq 0.153$. The latter ratio indicates that only a relatively small number of ratings from each client are included in the data (their full shipment history is not included in the sampled data). The data are not dominated by a single row or column because $\epsilon_R\doteq 9\times 10^{-6}$ and $\epsilon_C \doteq 0.0143$. Similarly 
\begin{align*}
\frac{N}{\sum_i\nid^2} & \doteq 0.103, && \frac{\sum_i\nid^2}{N^2} \doteq 1.95\times 10^{-6},\\
\frac{N}{\sum_j\ndj^2} & \doteq 1.22\times 10^{-4},\quad\text{and}&& \frac{\sum_j\ndj^2}{N^2} \doteq 0.00164. 
\end{align*}

Our two-factor linear mixed effects model for this data is:
\begin{model}\label{mod:stitchfixmodel}
For client $i$ and item $j$,
\begin{align}
\text{rating}_{ij} & =  \beta_0+\beta_1\text{match}_{ij}+\beta_2\mathbb{I}\{\text{client edgy}\}_i+\beta_3\mathbb{I}\{\text{item edgy}\}_j \notag \\
&\phe +  \beta_4\mathbb{I}\{\text{client edgy}\}_i*\mathbb{I}\{\text{item edgy}\}_j+\beta_5\mathbb{I}\{\text{client boho}\}_i \notag \\
&\phe + \beta_6\mathbb{I}\{\text{item boho}\}_j+\beta_7\mathbb{I}\{\text{client boho}\}_i*\mathbb{I}\{\text{item boho}\}_j \notag \\
&\phe + \beta_8\text{material}_{ij}+a_i+b_j+e_{ij}. \notag 
\end{align}
Here $\text{material}_{ij}$ is a categorical variable that is implemented via indicator variables for each type of material. We chose `Polyester', the most common material, to be the baseline. 
\end{model}

In a regression analysis, Model~\ref{mod:stitchfixmodel} would be only one
of many models one might consider.  There would be numerous ways to encode
the variables, and the coefficients in any one model would depend on which other
variables were included. The odds of settling on exactly this model
are low. To keep the focus on estimated standard errors due
to variance components we will work with a naive face-value interpretation
of the coefficients $\beta_j$ in Model~\ref{mod:stitchfixmodel}.
If the emphasis is on prediction, then one can use $\xij^\tran\betah$
perhaps adding shrunken row and/or column means of
the residuals. \cite{GO17} consider how estimates of $\ssa$, $\ssb$, and
$\sse$ can be used to shrink row and/or column means.

\begin{table}[t]
\caption{Stitch Fix Regression Results}
\label{tab:stitchfixresults}
\begin{tabular}{lllllll}  
\toprule 
& \multicolumn{1}{c}{$\betah_{\ols}$} & \multicolumn{1}{c}{$\wh{\mathrm{se}}_{\ols}(\betah_{\ols})$} & $\wh{\mathrm{se}}(\betah_{\ols})$ & \multicolumn{1}{c}{$\betah$} & \multicolumn{1}{c}{$\wh{\mathrm{se}}(\betah)$} \\
\midrule 
Intercept & $\phm 4.635^*$ & $0.005397$ & $0.05808$ & $\phm 5.110^*$ & $0.01250$ \\
	Match & $\phm5.048^*$ & $0.01174$ & $0.1464$ & $\phm 3.529^*$ & $0.02153$ \\
	$\mathbb{I}\{\text{client edgy}\}$ & $\phm 0.001020$ & $0.002443$ & $0.004593$ & $\phm 0.001860$ & $0.003831$ \\
	$\mathbb{I}\{\text{item edgy}\}$ & $-0.3358^*$ & $0.004253$ & $0.03730$ & $-0.3328^*$ & $0.01542$ \\
	$\mathbb{I}\{\text{client edgy}\}$ \\ $\hspace{5mm}{*}\mathbb{I}\{\text{item edgy}\}$ & $\phm 0.3925^*$ & $0.006229$ & $0.01352$ & $\phm 0.3864^*$ & $0.006432$ \\
	$\mathbb{I}\{\text{client boho}\}$ & $\phm 0.1386^*$ & $0.002264$ & $0.004354$ & $\phm 0.1334^*$ & $0.003622$ \\
	$\mathbb{I}\{\text{item boho}\}$ & $-0.5499^*$ & $0.005981$ & $0.03049$ & $-0.6261^*$ & $0.01661$ \\
	$\mathbb{I}\{\text{client boho}\}$ \\ $\hspace{5mm}*\mathbb{I}\{\text{item boho}\}$ & $\phm 0.3822^*$ & $0.007566$ & $0.01057$ & $\phm 0.3837^*$ & $0.007697$ \\
	Acrylic & $-0.06482^*$ & $0.003778$ & $0.03804$ & $-0.01627$ & $0.02149$ \\
	Angora & $-0.01262$ & $0.007848$ & $0.09631$ & $\phm 0.07271$ & $0.05837$ \\
	Bamboo & $-0.04593$ & $0.06215$ & $0.2437$ & $\phm 0.05420$ & $0.1716$ \\
	Cashmere & $-0.1955^*$ & $0.02484$ & $0.1593$ & $\phm 0.01354$ & $0.1176$ \\
	Cotton & $\phm 0.1752^*$ & $0.003172$ & $0.04766$ & $\phm 0.09743^*$ & $0.01811$ \\
	Cupro & $\phm0.5979^*$ & $0.3016$ & $0.4857$ & $\phm 0.5603$ & $0.4852$ \\
	Faux Fur & $\phm 0.2759^*$ & $0.02008$ & $0.08631$ & $\phm 0.3649^*$ & $0.07524$ \\
	Fur & $-0.2021^*$ & $0.03121$ & $0.1560$ & $-0.03478$ & $0.1331$ \\
	Leather & $\phm 0.2677^*$ & $0.02482$ & $0.08671$ & $\phm 0.2798^*$ & $0.07335$ \\
	Linen & $-0.3844^*$ & $0.05632$ & $0.2729$ & $\phm 0.006269$ & $0.1660$ \\
	Modal & $\phm 0.002587$ & $0.009775$ & $0.2052$ & $\phm 0.1417^*$ & $0.06498$ \\
	Nylon & $\phm 0.03349^*$ & $0.01552$ & $0.1000$ & $\phm 0.1186$ & $0.06436$ \\
	Patent Leather & $-0.2359$ & $0.1800$ & $0.4235$ & $-0.2473$ & $0.4222$ \\
	Pleather & $\phm 0.4163^*$ & $0.008916$ & $0.09905$ & $\phm 0.3344^*$ & $0.05023$ \\
	PU & $\phm 0.4160^*$ & $0.008225$ & $0.09019$ & $\phm 0.4951^*$ & $0.04196$ \\
	PVC & $\phm 0.6574^*$ & $0.06545$ & $0.3898$ & $\phm 0.8713^*$ & $0.3883$ \\
	Rayon & $-0.01109^*$ & $0.002951$ & $0.04602$ & $\phm 0.01029$ & $0.01493$ \\
	Silk & $-0.1422^*$ & $0.01317$ & $0.1004$ & $-0.1656^*$ & $0.05471$ \\
	Spandex & $\phm 0.3916^*$ & $0.01729$ & $0.1549$ & $\phm 0.3631^*$ & $0.1284$ \\
	Tencel & $\phm 0.4966^*$ & $0.009313$ & $0.1935$ & $\phm 0.1548^*$ & $0.06718$ \\ 
	Viscose & $\phm 0.04066^*$ & $0.006953$ & $0.09620$ & $-0.01389$ & $0.03527$ \\
	Wool & $-0.06021^*$ & $0.006611$ & $0.08141$ & $-0.006051$ & $0.03737$ \\
\bottomrule 
\end{tabular}
\end{table}

Suppose that one ignored client and item random effects and simply ran OLS. 
The resulting reported regression coefficients and standard errors are shown in the first two columns of Table~\ref{tab:stitchfixresults}. 
Estimated coefficients are starred if they would have been 
reported as being significant at the $0.05$ level. 
The third column has more realistic standard errors of the OLS regression coefficients, accounting for both the client and item random effects. These standard errors were computed using the variance component estimates from our algorithm
as described in Section~\ref{sec:varbetahrls}. As expected, they can be much larger, often by a factor of ten, than the OLS reported standard errors. 
A ten-fold increase in standard error corresponds to a hundred-fold
decrease in effective sample size.

In our simple model,
ten of the variables, `Acrylic', `Cashmere', `Cupro', `Fur', `Linen', `PVC', `Rayon', `Silk', `Viscose', and `Wool', that appear significantly different from polyester by OLS are not really significant when one accounts for client and item correlations.
An OLS analysis would lead to decisions being made with misplaced confidence.
It is likely that industry uses more elaborate models than our simple
regression, but a lower than anticipated effective sample size will
remain an issue.

The final two columns contain the regression coefficients estimated by our algorithm and their standard errors as defined in Section~\ref{sec:betaclt}. Again, estimated coefficients are starred if they are significant at the $0.05$ level. The estimated variance components are $\ssah=1.133$, $\ssbh=0.1463$, and $\sseh=4.474$. Their standard errors are approximately
$0.0046$, $0.00089$, and $0.0050$ respectively, so these components are well determined.
The error variance component is largest, and the client effect dominates the item effect
by almost a factor of eight.

The `Match' variable is significantly positively associated with rating, indicating that the baseline prediction provided by Stitch Fix is a useful predictor in this data set. However the random effects model reduces its coefficient from about $5$ to about $3.5$, a change that is quite a large number of estimated standard errors.  We have seen that some clients tend to give higher ratings on average than others. That is, client indicator variables take away some of the explanatory power of the match variable. 


Shipping an edgy item to a client who does not like edgy styles is associated with a rating decrease of about $0.33$ points, but shipping such an item to a client who does like edgy styles is associated with a small increase in rating. 

The boho indicator variable also has a negative overall 
estimated coefficient $\hat\beta_6<0$.
The modeled impact of a boho item sent to a boho client is
$\hat\beta_5+\hat\beta_6+\hat\beta_7<0$, unlike the positive
result we saw for sending and edgy item to an edgy client.
This suggests that it is more difficult to make matches for                         
boho items.  Perhaps there is an interaction where `boho to boho'
has a positive impact for a sufficiently high value of the match variable.
For large data sets, such an interaction can be conveniently
handled by filtering the data to cases with Match$_{ij}\ge t$
and refitting.  We did so but did not find a threshold
that yielded $\hat\beta_6+\hat\beta_5+\hat\beta_7 >0$.

Of the materials, `Cotton', `Faux Fur', `Leather', `Modal', `Pleather', `PU', `PVC', `Silk', `Spandex', and `Tencel' are significantly different from the baseline, `Polyester in our crossed random effects model. `PU' and `PVC' are associated with an increase in rating of at least half a point. Those materials are often used to make shoes and specialty clothing, which may be related to their association with high ratings. 

The computations in this section were done in python.

\section{Asymptotic behavior}\label{sec:theory}

Here we give sufficient conditions to ensure that the parameter estimates $\betah$, $\ssah$, $\ssbh$, and $\sseh$ obtained from Algorithm~\ref{alg:alternating} 
are consistent.  We also give a central limit theorem for $\betah$. 
We use the sample size growth conditions from Section~\ref{sec:background}
and some additional conditions on $\xij$.
Our results are conditional on the observed predictors $\xij$ for which $\zij=1$. 


As in ordinary IID error regression problems 
our central limit theorem
requires the information in the observed $\xij$ to
grow quickly in every projection while also imposing
a limit on the largest $\xij$.
For each $i$ with $\nid>0$, let $\xbid$ be the average of those
$\xij$ with $\zij=1$ and similarly define column
averages~$\xbdj$.

For a symmetric positive semi-definite matrix $V$, let $\info(V)$
be the smallest eigenvalue of $V$.
We will need lower bounds on  $\info(V)$
for various $V$ to rule out singular or nearly singular designs.  
Some of those $V$ involve centered variables.
In most applications $\xij$ will include
an intercept term, and so we assume that the first component
of every $\xij$ equals $1$.
That term raises some technical difficulties as centering that
component always yields zero.
We will treat that term specially in some of our proofs.
For a symmetric matrix $V\in\real^{p\times p}$, we let
$$\info_0(V) = \info( (V_{ij})_{2\le i,j\le p})$$ 
be the smallest eigenvalue of the lower 
$(p-1)\times (p-1)$ submatrix of $V$.

In our motivating applications, it is  reasonable to assume that $\Vert\xij\Vert$
are uniformly bounded.  We let
\begin{align}\label{eq:defmn}
M_N\equiv \max_{ij} \zij\Vert\xij\Vert^2
\end{align}
quantify the largest $\xij$ in the data so far.
Some of our results would still
hold if we were to let $M_N$ grow slowly with $N$. To
focus on the essential ideas, we simply take $M_N\le M_\infty <\infty$ for all $N$.

\subsection{Consistency}\label{sec:consistent}

First, we give conditions under which $\betah_{\ols}$ from step 1
is consistent.
\begin{thm}\label{thm:olscons}
Let $\max(\epsilon_R,\epsilon_C)\to 0$ and $\info(X^\tran X)\ge cN$ for some $c>0$,
as $N\to\infty$.
Then $\e(\Vert\betah_{\ols}-\beta\Vert^2)=O((\epsilon_R+\epsilon_C)/\info(X^\tran X/N)) \to0$.
\end{thm}
\begin{proof}
See Section~\ref{sec:proofolscons} in the supplement.
\end{proof}

Second, we show that the variance component estimates computed in step 2
are consistent. 
Recall that we compute the $U$-statistics \eqref{eq:Ustats} on data $(i,j,\hat\eta_{ij}=\yij-\xij^\tran\betah)$ and use them to obtain estimates $\ssah$, $\ssbh$, and $\sseh$ via \eqref{eq:mom}. 


\begin{thm}\label{thm:momcons}
Suppose that as $N\to\infty$
that $\max(\epsilon_R,\epsilon_C)\to0$,
$\max(R,C)/N\le\theta\in(0,1)$,  $\betah \p \beta$, 
and that $M_N$ is bounded. Then
$\ssah \p \ssa$, $\ssbh \p \ssb$, and $\sseh \p \sse$.
\end{thm}
\begin{proof}
See Section~\ref{sec:proofmomcons} in the supplement.
\end{proof}

From Theorem~\ref{thm:olscons}, the estimate of $\beta$ obtained in step $1$ of Algorithm~\ref{alg:alternating} is consistent. 
Therefore, from Theorem~\ref{thm:momcons}, the variance component estimates obtained in step $2$ are consistent, 
given the combined assumptions of those two theorems.
The proof of Theorem~\ref{thm:olscons} shows that 
the estimated  variance components differ by
 $O(\Vert\betah-\beta\Vert^2 +\epsilon \Vert\betah-\beta\Vert)$
from what we would get replacing $\betah$ by an oracle value $\beta$
and computing variance components of $\yij-\xij^\tran\beta$.
Such an estimate would have mean squared error $O(\epsilon)$
by Theorem~\ref{thm:varcoefbiasandvar}. As a result
the mean squared error for all parameters of interest is
$O(\epsilon)$.

Our third result  shows that the estimate of $\beta$ obtained in step $3$ is consistent. We do so by showing that estimators $\betah_{\rls}$ and $\betah_{\cls}$ are consistent when constructed using consistent variance component estimates. We give the version for  $\betah_{\rls}$.

\begin{thm}\label{thm:rlscons}
Let $\betah_{\rls}$ 
be computed with $\ssah \p \ssa$
and $\sseh \p \sse$ as $N\to\infty$,
where $\sse>0$.
If $\max(\epsilon_R,\epsilon_C)\to0$ and, 
\begin{align}\label{eq:ctrdsumsamplecov}
\info_0\Bigl(\sum_{ij} \zij (\xij-\xbid) (\xij-\xbid)^\tran/N\Bigr)\ge c>0
\end{align}
and
\begin{align}\label{eq:effnumcols}
\frac1{R^2}\sum_{ir}\zzt_{ir}\nid^{-1}\nrd^{-1}\to 0,
\end{align}
then $\betah_{\rls} \p \beta$.
\end{thm}
\begin{proof}
See Section~\ref{sec:proofrlscons} in the supplement.
\end{proof}

The most complicated part of the proof of Theorem~\ref{thm:rlscons} involves handling the
contribution of $\bj$ to $\betah_{\rls}$. In row weighted GLS it is quite standard to have
random errors $\ai$ and $\eij$ but here we must also contend with errors $\bj$
that do not appear in the model for which $\betah_{\rls}$ is the MLE.
Condition~\eqref{eq:effnumcols} is used to control the variance
contribution of the column random effects to the intercept
in $\betah_{\rls}$.  For balanced data it reduces
to $1/C\to0$ and so it has an effective number of columns
interpretation.  Recalling that $\zzt_{ir}$ is the number of
columns sampled in both rows $i$ and $r$, we have 
$\zzt_{ir}\le\nrd$ and so a sufficient condition for~\eqref{eq:effnumcols}
is that $(1/R)\sum_i\nid^{-1}\to0$.  For sparsely observed
data we expect $\zzt_{ir}\ll \max(\nid,\nrd)$ to be typical,
and then these bounds are conservative.

Any realistic setting will have $\sse>0$ and 
we need $\sse>0$ for $\betah_{\rls}$ to be well defined.
So that condition in Theorem~\ref{thm:rlscons} is not restrictive.

It remains to show that the variance component estimates from step 4 
are consistent.  We can just apply
Theorem~\ref{thm:momcons} again.
Therefore the final estimates returned by Algorithm~\ref{alg:alternating} are consistent given only weak conditions on the behavior of $\zij$ and $\xij$.

\subsection{Asymptotic Normality of $\betah$}\label{sec:betaclt}\label{sec:clt}


Here we show that the estimator $\betah_{\rls}$ constructed using consistent estimates of $\ssa$, $\ssb$, and $\sse$ 
is asymptotically Gaussian.
We need stronger conditions than we needed for consistency.

These conditions are expressed in terms of some 
weighted means of the predictors. First, let
\begin{align}\label{eq:xtdj}
\xtdj  =\frac1{\ndj}
\sum_i\zij\frac{\ssa}{\ssa+\sse/\nid}\xbid. 
\end{align}
This is a `second order' average of $x$ for column $j$:
it is the average over rows $i$
that intersect $j$, of averages $\xbid$ shrunken towards zero. 
For a balanced design with $\zij=1_{i\le R}1_{j\le C}$ we 
would have $\xtdj = \xbdd\ssa/(\ssa+\sse/C)$, so then the
second order means would all be very close to $\xbdd$ for large $C$.
Apart from the shrinkage, we can think of $\xtdj$ as a local version
of $\xbdd$ appropriate to column $j$.
Next let
\begin{align}\label{eq:R2k}
k = \sum_j\ndj^2(\xbdj-\xtdj)\Bigm/\sum_j\ndj^2\in\real^p.
\end{align}
This is a weighted sum of adjusted columns means, weighted
by the squared column size.  The intercept component of this $k$ will not be used.



\begin{thm}\label{thm:rlsnormal}
Let $\betah_{\rls}$ be computed with $\ssah \p \ssa$, $\ssbh \p \ssb$, and $\sseh \p \sse>0$ as $N\to\infty$.
Suppose also that 
\begin{align*}
&\info\biggl( \sum_i\xbid\xbid^\tran\biggr),\quad
\info_0\biggl( \sum_{ij}\zij(\xij-\xbid)(\xij-\xbid)^\tran\biggr),\quad\text{and}\\
&
\info_0\biggl(\sum_j\ndj^2(\xbdj-\xtdj-k)(\xbdj-\xtdj-k)^\tran\biggr)
\Bigm/\max_j\ndj^2
\end{align*}
all tend to infinity, where $\xtdj$ is given by~\eqref{eq:xtdj} and
$k$ is given by~\eqref{eq:R2k}.
Next for $c_j = \sum_i\zij\sse/(\sse+\ssa\nid)$ and $c_{ij} = \sse/(\sse+\ssa\nid)$
assume that both $\max_jc_j^2/\sum_jc_j^2$
and $\max_{ij}c_{ij}^2/\sum_{ij}c_{ij}^2$ converge to zero.
Then
$\betah_{\rls}$ is asymptotically distributed as 
\begin{align}\label{eq:rlsnormal}
\dnorm(\beta, (X^\tran V_A^{-1} X)^{-1} X^\tran V_A^{-1} V_R V_A^{-1}X(X^\tran V_A^{-1}X)^{-1}).
\end{align}
\end{thm}
\begin{proof}
See Section~\ref{sec:proofrlsnormal} in the supplement.
\end{proof}

The statement that $\betah_{\rls}$ has asymptotic distribution
$\dnorm(\beta,V)$ 
is shorthand for  $V^{-1/2}(\betah-\beta)\p\dnorm(0,I_p)$.

Theorem~\ref{thm:rlsnormal} imposes three information criteria.
First, the $R$ rows $i$ with $\nid>0$ must have sample average $\xbid$ vectors
with information tending to infinity.  
It would be reasonable to expect that information to be proportional to $R$
and also reasonable to require $R\to\infty$ for a CLT.
Next, the sum of within row sums of squares and cross products of
row-centered $\xij$ must have growing information, apart from the intercept term.
Finally,  thinking of $\xbdj-\xtdj$ as locally centered mean for column $j$,
those quantities centered on the vector $k$
 must have a weighted sum of squares that is not dominated by any single column
when weights proportional to $\ndj^2$ are applied.

The conditions on $c_j$ and $c_{ij}$ are used to show that the CLT will apply
to the intercept in the regression.
The condition on $\max_jc^2_j/\sum_jc_j^2$ will fail if for example column $j=1$
has half of the $N$ observations, all in rows of size $\nid=1$. In the case of an $R\times C$
grid $\max_jc_j^2/\sum_jc_j^2 = 1/C$ and so we can interpret this condition as requiring 
a large enough effective number $\sum_jc_j^2/\max_jc_j^2$ of columns in the data. 

The condition on $\max_{ij}c_{ij}^2/\sum_{ij}c_{ij}^2$ will fail if for example the data contain
a full $R\times C$ grid of values plus a single observation with $i=R+1$ and $j=C+1$.
The problem is that in a row based regression, a single small row can get outsized leverage.
It can be controlled by dropping relatively small rows. This pruning of rows is only
used for the  CLT to apply to the intercept term. It is not needed for other components
of $\beta$ nor is it need for consistency. We do not know if it is necessary for the CLT.

\subsection{Computing $\var(\betah_\rls)$}\label{sec:varbetahrls}
Here we show  how to compute the estimate of the asymptotic variance of $\betah_{\rls}$ from Theorem~\ref{thm:rlsnormal}. First,
\begin{align}\label{eq:varbeta}
&(X^\tran V_A^{-1} X)^{-1} X^\tran V_A^{-1} V_R V_A^{-1}X(X^\tran V_A^{-1}X)^{-1}\notag\\
&= (X^\tran V_A^{-1} X)^{-1} X^\tran V_A^{-1} (V_A+\ssb B_R)V_A^{-1}X(X^\tran V_A^{-1}X)^{-1} \notag\\
&= (X^\tran V_A^{-1} X)^{-1} + (X^\tran V_A^{-1} X)^{-1} X^\tran V_A^{-1} \ssb B_R V_A^{-1}X(X^\tran V_A^{-1}X)^{-1} \notag\\
&= (X^\tran V_A^{-1} X)^{-1}+(X^\tran V_A^{-1} X)^{-1} \var(X^\tran V_A^{-1}{b})(X^\tran V_A^{-1}X)^{-1},
\end{align}
where ${b}$ is the length-$N$ vector of column random effects for each observation.
That is $\bj$ appears $\ndj$ times in $b$.

Using the Woodbury formula 
we find that $\var(X^\tran V_A^{-1}{b})$ equals
\begin{align}\label{eq:varbetapart2}
\dfrac{\ssb}{\fpe}\sum_j \Bigl(X_{\sumdot j}-\ssa\sum_i\zij\dfrac{X_{i\sumdot}}{\sse+\ssa\nid}\Bigr)\Bigl(X_{\sumdot j}-\ssa\sum_i\zij\dfrac{X_{i\sumdot}}{\sse+\ssa\nid}\Bigr)^\tran.
\end{align}
Recall that $X_{i\sumdot}$ and $X_{\sumdot j}$ are row and column 
totals, not means.

In practice, we plug consistent estimates $\ssah$, $\ssbh$, and $\sseh$ in for $\ssa$, $\ssb$, and $\sse$ in \eqref{eq:varbeta} and \eqref{eq:varbetapart2}. We already have $(X^\tran \hat V_A^{-1}X)^{-1}$ as well as $\nid$ and $X_{i\sumdot}$ for $i=1,\dots,R$ available from computing $\betah_{\rls}$. In a new pass over the data, we compute $X_{\sumdot j}$ and $\sum_i\zij X_{i\sumdot}$ for $j=1,\dots,C$, incurring $O(Np)$ computational and $O(Cp)$ storage cost. Then, \eqref{eq:varbetapart2} can be found in $O(Cp^2)$ time; a final step finds \eqref{eq:varbeta} in $O(p^3)$ time. Overall, estimating the variance of $\betah_{\rls}$ requires $O(Np+Cp^2+p^3)$ additional computation time and $O(Cp)$ additional space.

\section{Comparisons to the MLE}\label{sec:experiments}

Here we compare Algorithm~\ref{alg:alternating} to maximum likelihood
for a linear mixed effects model, looking at both
computational efficiency and statistical efficiency.
We use a state of the art code for linear mixed models called
 MixedModels \citep{lme4julia}.
This is written in Julia \citep{beza:edel:karp:shah:2017}
and is much faster than other linear
mixed model code we have tried.

Our examples have $R=C=2\sqrt{N}$ for various $N$.
We create an $R\times C$ matrix of $\zij$ and randomly choose
exactly $RC/4$ components to be $1$.
We have an intercept and $p$ other $x$'s with
$x_{ij,t}\simiid\dnorm(0,1)$, for $2\le t\le p$.
We use all $p\in\{1,5,10,20\}$.
We take $\ssa=2$, $\ssb=1/2$, $\sse=1$
and all $\beta_j=1$.
Our simulated random effects and our noise are all Gaussian because
we are comparing to code that computes a Gaussian MLE.

\subsection{Computational cost}

The MixedModels package in Julia uses a derivative-free optimization
method from the BOBYQA package \citep{powe:2009}.  At each iteration it 
evaluates the log likelihood at a set of points, fits a quadratic function
to those points and minimizes the quadratic.  
The number of  likelihood evaluations per iteration is fixed, 
but we are unable to model the number of iterations required. 
We consider the cost per likelihood evaluation next.

The log likelihood is
\begin{align*}
(Y-X\beta)^\tran(\ssa A_R+\ssb B_R+\sse I_N)^{-1}(Y-X\beta)
+\ln|\ssa A_R+\ssb B_R+\sse I_N|.
\end{align*}
In an analysis using the Woodbury formula we find that
the log likelihood can be computed in $O(R^3+\sum_i\nid^2)$ time.
Because $1\le R\le N$ we can write $R=N^\alpha$ for some $0\le\alpha\le 1$.
Then
$$
R^3+\sum_i\nid^2
=R^3+R\Bigl(\frac1R\sum_i\nid\Bigr)^2
\ge R^3+N^2R^{-1}=N^{3\alpha}+N^{2-\alpha}.
$$
Now $N^{3\alpha}+N^{2-\alpha}
>\max(N^{3\alpha},N^{2-\alpha})$ and $\alpha=1/2$
minimizes $\max(3\alpha,2-\alpha)$.
Therefore $R^3+\sum_i\nid^2>N^{3/2}$.

This is the same estimate one gets by
considering the cost of solving a system of $R+C$
equations in $R+C$ unknowns.
There are faster ways to solve the equations
in special cases like nested models, and there is the
possibility that sparsity patterns in the data can be
exploited for speed.  However, we are interested
in arbitrarily complicated sampling plans where
these special cases cannot be assumed.

Figure~\ref{fig:timeperiteration} shows
computed cost per iteration for $10$
replicates at each of $11$ different sample sizes $R^2/4$
given by $R=10$, $20$, $40$, $\dots$, $2^9\times 10$, with $p=5$.
The cost per iteration is flat for small $N$ presumably
due to some overhead.  It grows slowly until
about $N=10^4$ and then it appears to increase parallel
to a reference line with slope $3/2$.  

\begin{figure}
	\centering 
	\includegraphics[scale=0.65]{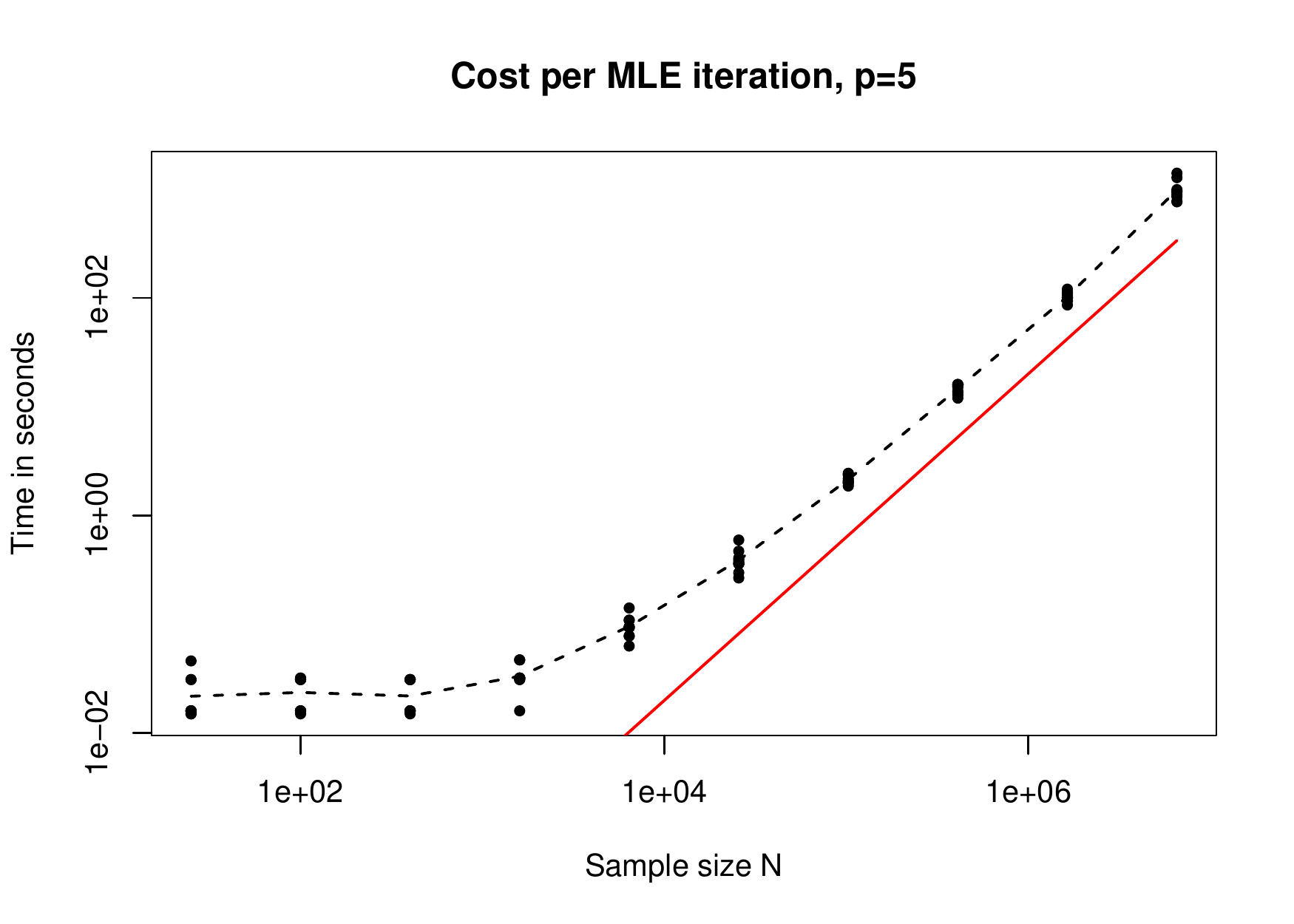}
	\caption{For $p=5$, cost per iteraton of MLE. 
The dashed curve is the average of $10$ replicates. 
The solid reference line is parallel to $N^{3/2}$.}
	\label{fig:timeperiteration}
\end{figure}

Figure~\ref{fig:cost5} shows total cost versus
$N$ in a setting with $p=5$ averaged over $100$ data sets.
The cost curve for the MLE computation looks different
from Figure~\ref{fig:timeperiteration}. It does not
start out flat for small $N$.  We found that the number of iterations
required to find the MLE generally rose over the range $64\le N\le 6400$
and then declined gently thereafter.

From the analysis and empirical results, we
find that a cost per iteration of $O(N^{3/2})$
is a realistic lower bound for the MLE code.
The method of moments cost is $O(N)$ theoretically and appears to 
be proportional to $N$ empirically.

\begin{figure}
	\centering 
	\includegraphics[scale=0.65]{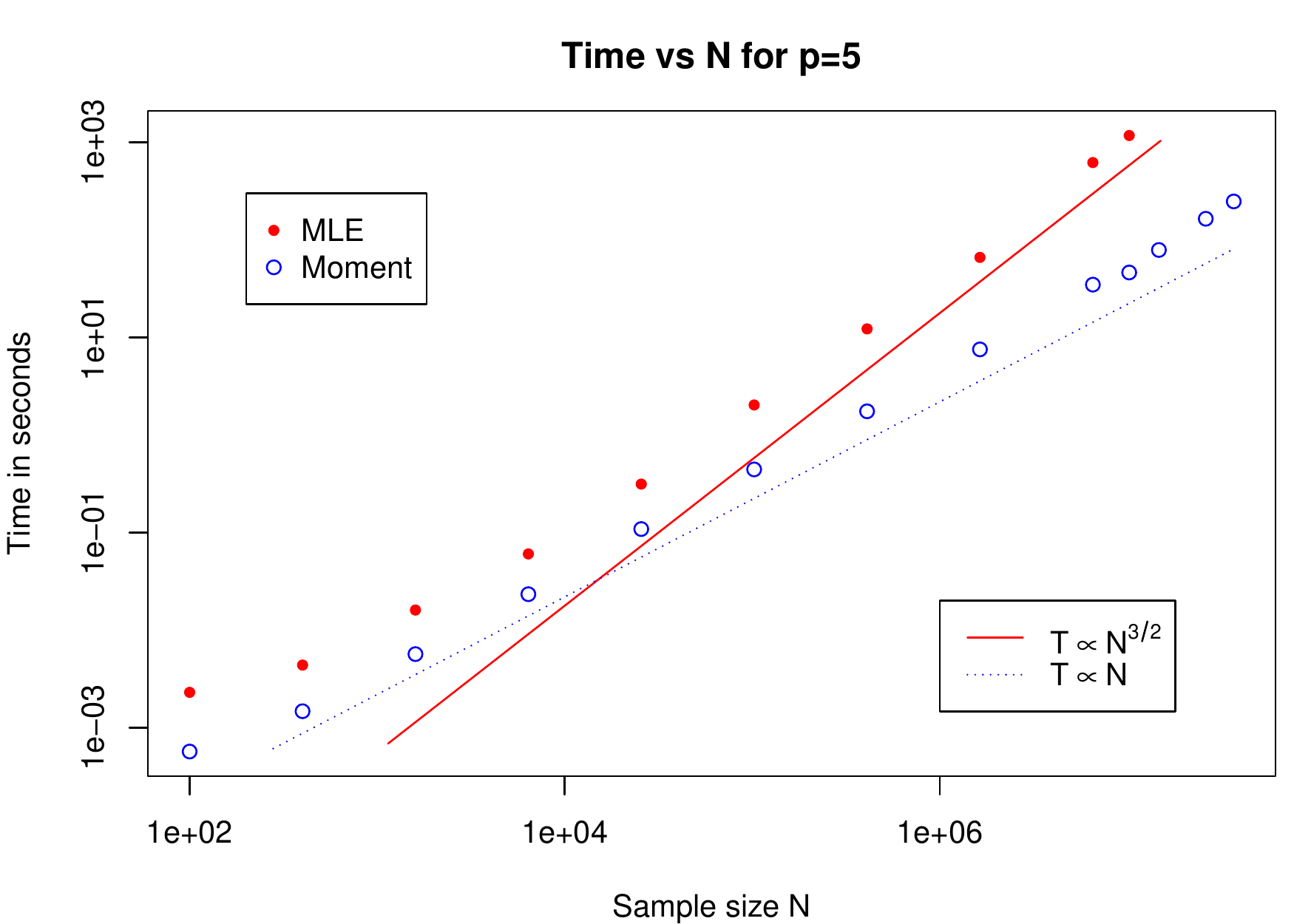}
	\caption{
Computational cost for MLE and moments versus sample size $N$. 
There are reference lines parallel to $N^{3/2}$ and $N^1$.}
	\label{fig:cost5}
\end{figure}

Our computations were done with data generated in memory.
In commercial applications, there could be a much larger
time cost proportional to $N$ involved in reading the data from external storage.
However, the $N^{3/2}$ cost component would be considerably
larger at commercial scale, where $N$ is much larger than
in our examples. For the method of moments it is straightforward
to read and use the data in parallel even for large $N$.

A second computational issue arises with the linear
mixed effects MLE. The code crashes on large enough data
sets because the algorithm requires $O((R+C)^2)$ memory.
For $p=5$ we were unable to take the next step
past $N=5120^2/4\doteq 6.5\times 10^6$. The program runs out of memory on our cluster.
For $p=1$, it crashes for $N$ near $18$ million observations.
The method of moments in Algorithm~\ref{alg:alternating}
has linear cost both theoretically and empirically and can be implemented in $O(R+C)$ memory.
The difference is minor for our CPU time simulations that also keep all $N$
observations in memory, but it will be critical in large commercial applications.

 

\subsection{Statistical efficiency}

For statistical efficiency we considered
$p=1$, $5$, $10$ and $20$.  
Sample sizes $N=100\times 4^j$ for $j=0,1,\dots,8$
were replicated $100$ times each.  A few larger values of $N$
were replicated $10$ times each, though the MLE code
would not run on all of the largest sample sizes we tried.
The pattern in the results was the same
for all of those $p$.  We display results for $p=5$
in Figure~\ref{fig:acc5}.  The MSEs for $\beta$  decay proportionally to $1/N$.  
The reference curves for variance components in Figure~\ref{fig:acc5} are
what we would expect from IID sampling of $\ai$, $\bj$ and $\eij$, respectively
namely $2\sigma_A^4/R$, $2\sigma_B^4/C$ and  $2\sigma_E^4/N$
where $R=C=2\sqrt{N}$.

\begin{figure}[t!]
	\centering 
	\includegraphics[scale=0.65]{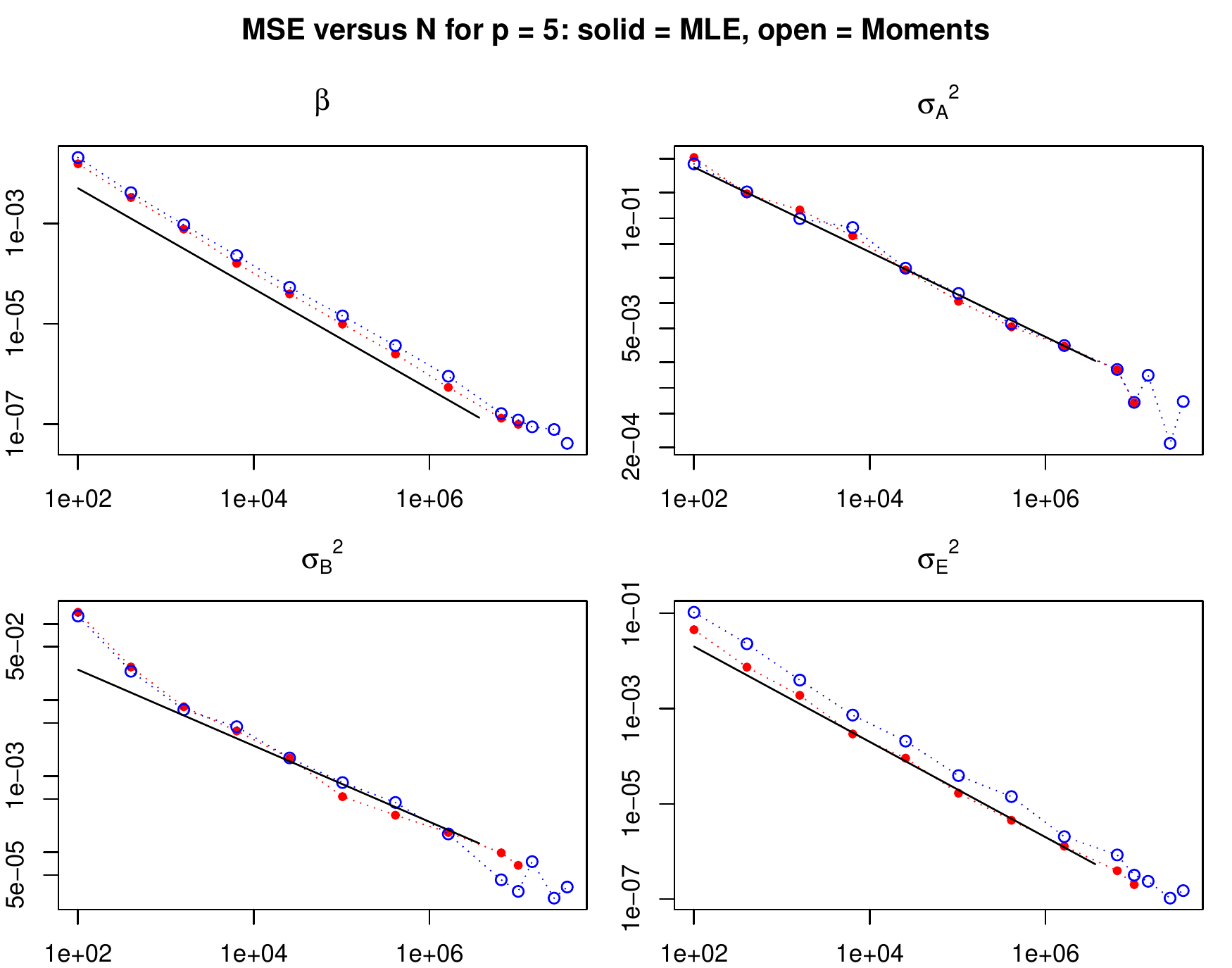}
	\caption{
Mean squared errors for $\beta$, $\ssa$, $\ssb$ and $\sse$
versus $N$. Reference lines for $\beta$ and $\sse$ are
parallel to $N^{-1}$. Reference lines for $\ssa$ and $\ssb$
are parallel to $N^{-1/2}$.
}
	\label{fig:acc5}
\end{figure}

The parameter of greatest interest will ordinarily
be $\beta$.  The MLE has greater accuracy for $\beta$,
as it must by the Gauss-Markov theorem. In this
instance the MLE has about half the MSE that the
method of moments does.
For the variance components,
the  method of moments attains essentially the same MSE
as the MLE does for $\ssa$ and $\ssb$.
The MLE has greater efficiency for $\sse$.  In ordinary
use we would want to know ratios of variance components
and the uncertainty in such ratios is dominated by that
in $\ssa$ and $\ssb$, where the two methods have comparable accuracy.

In this example, we saw a modest loss in statistical
efficiency of $\betah$ and $\sseh$ and comparable
accuracy for $\ssah$ and $\ssbh$.
These comparisons were run on data simulated from the Gaussian model that the MLE assumes. 
The method of moments does not require that assumption.  
Likelihood based variance estimates for variance components, such as $\wh\var(\ssah)$, can fail to be even asymptotically correct when the Gaussian model does not hold.

\section{Conclusion and Future Work}\label{sec:conclusion}

We have proposed an algorithm for the two-factor linear mixed effects model with crossed covariance structure that provides consistent and asymptotically normal parameter estimates. It alternates twice between estimating the regression coefficients and estimating the variance components via the method of moments.

Unlike the available methods based on Bayes theorem or maximum likelihood,
the moment estimates cost $O(N)$ time and $O(R+C)$ space.  The variance estimate for $\betah$
is obtained by substituting consistent estimates of $\ssa$, $\ssb$, and $\sse$
into exact finite sample formulas for that variance matrix.
The variance estimates for $\ssah$, $\ssbh$, and $\sseh$ are obtained by such a substitution
in mildly conservative formulas from \cite{GO17}.
Here the usual root $n$ consistency 
from IID settings is replaced by a $1/\sqrt{\epsilon}$
consistency for $\epsilon = \max(\epsilon_R,\epsilon_C)$.
Interpreting $1/\sqrt{\epsilon}$ as an effective sample size might be somewhat
conservative because in theorems such as Theorem~\ref{thm:varcoefbiasandvar}
the value of $\epsilon$ appears in  upper bounds.

We exchange higher MSEs for an algorithm with cost only linear in the number of observations.
We do not know how bad the efficiency loss might be in general, but we expect that
when the pure error term $\sse$ is meaningfully large that the loss will not be extreme. Also, if one of $\ssa$ and $\ssb$ very much dominates the other one, we can get a GLS estimate that accounts for the dominant source of correlation.

We anticipate that a martingale central limit theorem will
apply to the variance component estimates $\ssah$, $\ssbh$, and $\sseh$.
Some details will be in the forthcoming dissertation of the first author.
We do not anticipate those variance components to be uncorrelated
with $\betah$ because the random variables $\ai$, $\bj$, and~$\eij$
do not need to have symmetric distributions.

This paper is a second step in developing big data versions of mixed model
procedures such as the Henderson estimators.
One followup step is to incorporate interactions between fixed and random steps, as the Henderson III
model allows.  Another is to incorporate interactions among latent variables. At present both kinds of
interactions would serve to inflate $\sse$. A third step is to
adapt to  binary responses, for instance by replacing the identity link in Model~\eqref{eq:refmodel}, with a logit  or probit link. This third step is of value because many responses  in e-commerce are categorical. One binary response of interest to Stitch Fix is whether the client keeps the item of clothing. 

\section*{Acknowledgments}

This work was supported by the US NSF under grants DMS-1407397
and DMS-1521145.
KG was supported by US NSF Graduate Research Fellowship under grant DGE-114747. 
Any opinions, findings, and conclusions or recommendations expressed in this material are those of the 
authors and do not necessarily reflect the views of the National Science Foundation.

We would like to thank Stitch Fix and in particular Brad Klingenberg for providing us with the data used in our real-world experiment and motivation and encouragement during the project.

\begin{supplement}
  \sname{Supplement A}\label{SuppA}
  \stitle{Proofs}
\slink[url]{http://statweb.stanford.edu/$\sim$owen/reports/vllmemsupp.pdf}
  \sdescription{This supplement contains proofs of all the theorems and lemmas not proved in the main paper.}
\end{supplement}

\bibliographystyle{apalike}
\bibliography{MixedRefs}

\end{document}